\documentclass[11p,reqno]{amsart}

\usepackage[T1]{fontenc}
\usepackage{graphicx}
\usepackage{amssymb,amsthm,amsmath,mathrsfs,bm,braket,marginnote}
\usepackage{enumerate}
\usepackage{appendix}
\usepackage[colorlinks=true, pdfstartview=FitV, linkcolor=blue, citecolor=blue, urlcolor=blue]{hyperref}
\usepackage{multirow}

\usepackage{pgf}
\usepackage{pgfplots}
\usepackage{tikz}
\usetikzlibrary{arrows,calc}
\usepackage{verbatim}
\usetikzlibrary{decorations.pathreplacing,decorations.pathmorphing}
\usepackage[numbers,sort&compress]{natbib}
\usepackage{dsfont}

\newcommand{\tit}{{\tilde{\tau}}}
\newcommand{\mm}{{\mathcal{M}}}
\newcommand{\mn}{{\mathcal{N}}}
\newcommand{\ml}{{\mathcal{L}}}

\numberwithin{equation}{section}
\newtheorem{theorem}{Theorem}[section]

\newtheorem{remark}[theorem]{Remark}
\newtheorem{coro}[theorem]{Corollary}

\newtheorem{proposition}[theorem]{Proposition}

\begin{document}
\title[Cauchy-Jacobi orthogonal polynomials and the discrete CKP equation]{Cauchy-Jacobi orthogonal polynomials and the discrete CKP equation}

\date{}

\dedicatory{}


\author{Shi-Hao Li}
\address{Department of Mathematics, Sichuan University, Chengdu, 610064, PR China}
\email{shihao.li@scu.edu.cn}
\author{Satoshi Tsujimoto}
\address{Department of Applied Mathematics and Physics, Graduate School of Imformatics, Kyoto University, Yoshida-Honmachi, Kyoto, Japan 606-8501}
\email{tujimoto@i.kyoto-u.ac.jp}
\author{Ryoto Watanabe}
\address{Department of Applied Mathematics and Physics, Graduate School of Imformatics, Kyoto University, Yoshida-Honmachi, Kyoto, Japan 606-8501}
\email{watanabe.ryoto.g31@kyoto-u.jp}
\author{Guo-Fu Yu}
\address{School of Mathematical Sciences, Shanghai Jiaotong University, People's Republic of China; School of Mathematical Sciences, CMA-Shanghai, Shanghai Jiao Tong University, Shanghai 200240}
\email{gfyu@sjtu.edu.cn}

\begin{abstract}
\noindent 
This paper intends to construct discrete spectral transformations for Cauchy-Jacobi orthogonal polynomials, and find its corresponding discrete integrable systems. It turns out that the normalization factor of Cauchy-Jacobi orthogonal polynomials acts as the $\tau$-function of the discrete CKP equation, which has applications in Yang-Baxter equation, integrable geometry, cluster algebra, and so on. 
\end{abstract}

\maketitle

\section{Introduction}
The relation between orthogonal polynomials and integrable systems has been studied for many years, due to its deep connection with conformal field theory, random matrix theory and numerical algorithms, etc. On the one hand, orthogonal polynomials play crucial roles in the characterization of spectral problems, thus making connections with Lax representations for discrete integrable equations. 
On the other hand, new orthogonal polynomials were proposed when solving the spectral problem of  classical integrable systems. Cauchy bi-orthogonal polynomials is one of the examples proposed during the study of Hermite-Pad\'e approximation problem of Degasperis-Procesi peakon equation \cite{lundmark05}. Later on, classical properties of Cauchy bi-polynomials were investigated \cite{bertola09,bertola10}, including its  zeros distribution, four-term recurrence relation, Riemann-Hilbert problem and so on. Besides, such a polynomial is related to random matrix theory. In \cite{bertola14}, the Cauchy two-matrix model with specific Laguerre weight was considered, with limiting behavior expressed in terms of Meijer-G function. To be precise, the inner product of Cauchy-Laguerre weight is defined from $\mathbb{R}[x]\times\mathbb{R}[y]\to\mathbb{R}$, by
\begin{align*}
\langle f(x), g(y)\rangle=\int_{\mathbb{R}_+^2}\frac{x^ay^be^{-x-y}}{x+y}f(x)g(y)dxdy.
\end{align*}
To evaluate the related bi-moments $\langle x^i, y^j\rangle$ with Cauchy-Laguerre weight, we can introduce a one-parameter generalization of the inner product \cite[lemma 2.3]{forrester21}
\begin{align*}
\langle f(x),g(y)\rangle_t=\int_{\mathbb{R}_+^2}\frac{x^ay^be^{-t(x+y)}}{x+y}f(x)g(y)dxdy.
\end{align*}
This inner product plays a dominant role when we connect Cauchy bi-orthogonal polynomials with  Toda-type equation. In \cite{li19}, it has been shown that with $a=b$, the one-parameter deformed Cauchy-Laguerre weight is related to the CKP hierarchy, as well as a semi-discrete Toda equation of CKP type (C-Toda for brevity). 
If $a\neq b$, there is also a related discrete integrable system, with $\tau$-function expressed by an asymmetric Gram determinant \cite{chang21}. Besides, in \cite{miki11}, a full-discrete equation was constructed by introducing a discrete evolution of the inner product, while in \cite{li23}, a non-abelian C-Toda lattice was obtained by using matrix-valued Cauchy bi-orthogonal polynomials. 

In this paper, we mainly consider a deformed Cauchy bi-orthogonal polynomials with Jacobi weight, and its corresponding discrete integrable system. By introducing a two-parameter symmetric inner product
\begin{align*}
\langle f(x),g(y)\rangle_{s,t}=\int_{[0,1]\times[0,1]}\frac{x^sy^s}{x+y}f(x)g(y)\left(
\frac{1-x}{1+x}
\right)^t\left(
\frac{1-y}{1+y}
\right)^t dxdy,\quad s,t\in\mathbb{N},
\end{align*}
we show the corresponding orthogonal polynomials $\{P_n^{s,t}(x)\}_{n\in\mathbb{N}}$ not only have the four-term recurrence relation, but some spectral transformations with respect to parameters $s$ and $t$. Different with the discrete time evolutions considered in \cite{miki11}, we consider an evolution with respect to variables $x$ and $y$ simultaneously.

In Section \ref{sec2}, we construct different spectral transformations in $s$-direction and $t$-direction, which, compatible with four-term recurrence relation, leading to an algebraic identity and nonlinear difference equations. By noting the normalization factor could be simply expressed by a discrete Gram determinant, in Section \ref{sec3}, we simplify nonlinear equations by using spectral transformations, and decouple nonlinear equations into bilinear equations. It was finally found that these bilinear equations have also been derived in \cite{kenyon15} from hexahedron recurrence relations. Besides, the bilinear equations could be expressed in terms of a single quartic equation, governed by the normalization factor of Cauchy-Jacobi orthogonal polynomials. It is nothing but the discrete CKP equation, equivalent to the star-triangle move in the Ising model studied by Kashaev \cite{kashaev95}. 

To conclude, there are some highlights of this paper: (1) We introduce a two-parameter deformation for Cauchy orthogonal polynomials in Jacobi type, and study its related discrete spectral transformations; (2) We provide a determinant solution for the discrete CKP equation. Different from the soliton solution, this is a molecule solution determined by its initial values; (3) We construct a Lax pair for the discrete CKP equation. The integrability of discrete CKP equation was obtained by Yang-Baxter equation, and we show its Lax integrability by using orthogonal polynomial method.

\section{Two-parameter generalization of Cauchy orthogonal polynomials}	\label{sec2}
Let's consider an inner product with Cauchy weight by defining a symmetric bilinear form $\langle\cdot,\cdot\rangle: \mathbb{R}[x]\times\mathbb{R}[y]\to\mathbb{R}$
\begin{align*}
\langle f(x),g(y)\rangle=\int_{\mathbb{R}_+\times\mathbb{R}_+}\frac{1}{x+y}f(x)g(y)\omega(x)\omega(y)dxdy.
\end{align*}
Here we assume that the weight function $\omega(x)$ is supported on a subset of $\mathbb{R}_+$ to avoid the singularity brought by the Cauchy kernel $1/(x+y)$.
With this inner product, we could obtain a family of monic orthogonal polynomials $\{P_n(x)\}_{n\in\mathbb{N}}$ by requiring the orthogonality
\begin{align*}
\langle P_n(x),P_m(y)\rangle=h_n\delta_{n,m},
\end{align*}
where $\delta_{n,m}=0$ if $n\neq m$ and $\delta_{n,m}=1$ if $n=m$. 
In this paper, we focus on a two-parameter deformation of the Cauchy inner product, reading as
\begin{align*}
\langle f(x),g(y)\rangle_{s,t}=\int_{(0,1)\times(0,1)}\frac{x^sy^s}{x+y}f(x)g(y)\left(
\frac{1-x}{1+x}
\right)^t\left(
\frac{1-y}{1+y}
\right)^t dxdy,\quad s,t\in\mathbb{N}.
\end{align*}
In other words, we choose a deformed Jacobi weight function $\omega(x)=x^s(1-x)^t (1+x)^{-t}$, supported on $(0,1)$.
With this deformed inner product, we could introduce a family of monic two-parameter Cauchy orthogonal polynomials $\{P^{s,t}_n(x)\}_{n\in\mathbb{N}}$ by orthogonality $\langle P_n^{s,t}(x),P_m^{s,t}(y)\rangle_{s,t}=h_n^{s,t}\delta_{n,m}$, where $h_n^{s,t}$ is a non-zero constant. 
\begin{proposition}
The two-parameter deformed monic Cauchy orthogonal polynomials have the following determinant expression
\begin{align*}
P_n^{s,t}(x)=\frac{1}{\tau_n^{s,t}}\left|\begin{array}{cccc}
m_{0,0}^{s,t}&\cdots&m_{0,n-1}^{s,t}&1\\
\vdots&&\vdots&\vdots\\
m_{n-1,0}^{s,t}&\cdots&m_{n-1,n-1}^{s,t}&x^{n-1}\\
m_{n,0}^{s,t}&\cdots&m_{n,n-1}^{s,t}&x^n
\end{array}
\right|,\quad \tau_n^{s,t}=\det\left(
m_{i,j}^{s,t}
\right)_{i,j=0}^{n-1},
\end{align*}
where
\begin{align*}
m_{i,j}^{s,t}=\langle x^i,y^j\rangle_{s,t}=\int_{(0,1)^2}\frac{x^{s+i}y^{s+j}}{x+y}\left(
\frac{1-x}{1+x}
\right)^t\left(
\frac{1-y}{1+y}
\right)^tdxdy.
\end{align*}
\end{proposition}
Using the determinant formula, it is not difficult to verify that $h_n^{s,t}=\tau_{n+1}^{s,t}/\tau_n^{s,t}$, and the non-zero property of normalization factor is guaranteed by the Andr\'eief formula (for reference, see \cite[Equation (2.5)]{bertola09})
\begin{align}\label{multiple}
\tau_n^{s,t}=\int_{0<x_1<\cdots<x_n<1\atop 0<y_1<\cdots<y_n<1}\prod_{i,j=1}^n \frac{x_i^sy_j^s}{x_i+y_j}\prod_{1\leq i<j\leq n}(x_i-x_j)^2(y_i-y_j)^2\prod_{i=1}^n \left(
\frac{1-x_i}{1+x_i}
\right)^t\left(
\frac{1-y_i}{1+y_i}
\right)^t
dx_idy_i.
\end{align}
Noting that every term in the integration is positive within the interval $(0,1)$, we can define a probability measure on the configuration space. However, different from the Cauchy-Laguerre case, it is still unknown how to evaluate this partition function in terms of hypergeometric functions.

Moreover, we have the following properties of the deformed inner product.
\begin{proposition}
The two-parameter deformed inner product has the following evolution properties:
\begin{enumerate}
\item $
\langle f(x),g(y)\rangle_{s+1,t}=\langle xf(x),yg(y)\rangle_{s,t};
$
\item $\langle f(x),g(y)\rangle_{s,t+1}=\langle f(x),g(y)\rangle_{s,t}-\mathcal{L}_{s,t}(f)\mathcal{L}_{s,t}(g)$, where 
\begin{align*}
\mathcal{L}_{s,t}(f)=\sqrt{2}\int_0^1\frac{x^s}{1+x}f(x)\left(
\frac{1-x}{1+x}
\right)^tdx.
\end{align*}
\end{enumerate}
\end{proposition}
As a direct result, we know that moments $\{m_{i,j}^{s,t}\}_{i,j\in\mathbb{N}}$ satisfy relations
\begin{align}
\begin{aligned}
m_{i,j}^{s+1,t}=m_{i+1,j+1}^{s,t},\quad 
m_{i,j}^{s,t+1}=m_{i,j}^{s,t}-\phi_i^{s,t}\phi_j^{s,t}
\end{aligned}
\end{align}
with $\phi_i^{s,t}=\mathcal{L}_{s,t}(x^i)$. 
Similar to ordinary Cauchy orthogonal polynomials, we have the following four-term recurrence relation for $\{P_n^{s,t}(x)\}_{n\in\mathbb{N}}$.
\begin{proposition}\label{prop_4trr}
For any $s,\,t\in\mathbb{N}$, we have 
\begin{align}\label{4trr}
\begin{aligned}
x(P_n^{s,t}(x)&+a_n^{s,t}P_{n-1}^{s,t}(x))\\
&=P_{n+1}^{s,t}(x)+(a_n^{s,t}-b_n^{s,t})P_n^{s,t}(x)+(-c_n^{s,t}+a_n^{s,t}b_{n-1}^{s,t})P^{s,t}_{n-1}(x)-a_n^{s,t}c_{n-1}^{s,t}P^{s,t}_{n-2}(x)
\end{aligned}
\end{align}
for some non-zero coefficients
$a_n^{s,t}$, $b_n^{s,t}$ and $c_n^{s,t}$, 
where initial values are given by $P_{-1}^{s,t}(x)=0$ and $P_0^{s,t}(x)=1$.
\end{proposition}
It has been noticed in \cite{li19} that variables $a_n^{s,t}$, $b_n^{s,t}$ and $c_n^{s,t}$ have determinant expressions. By introducing
\begin{align}\label{sigma}
\tilde{\sigma}_{n}^{s,t}=\left|\begin{array}{cccc}
m^{s,t}_{0,0}&\cdots&m^{s,t}_{0,n-1}&m^{s,t}_0\\
m^{s,t}_{1,0}&\cdots&m^{s,t}_{1,n-1}&m^{s,t}_1\\
\vdots&&\vdots&\vdots\\
m^{s,t}_{n,0}&\cdots&m^{s,t}_{n,n-1}&m_n^{s,t}
\end{array}
\right|,\quad m_i^{s,t}=\int_0^1 x^{s+i} \left(
\frac{1-x}{1+x}
\right)^t dx
\end{align}
and denoting $p_n^{s,t}$ as the coefficient of $x^{n-1}$ in $P_n^{s,t}(x)$, we have
\begin{align}\label{coeff1}
a_n^{s,t}=-\frac{\tilde{\sigma}^{s,t}_{n}\tau^{s,t}_{n-1}}{\tilde{\sigma}^{s,t}_{n-1}\tau^{s,t}_{n}},
\quad b_n=p_{n+1}^{s,t}-p_n^{s,t},\quad c_n^{s,t}=\frac{\tau_{n-1}^{s,t}\tau_{n+1}^{s,t}}{(\tau_n^{s,t})^2}.
\end{align}
It should be remarked that $p_{n}^{s,t}$ has determinant expression as well. If we introduce
\begin{align*}
\tit_n^{s,t}=\left|\begin{array}{ccc}
m_{0,0}^{s,t}&\cdots&m_{0,n-1}^{s,t}\\
\vdots&&\vdots\\
m_{n-2,0}^{s,t}&\cdots&m_{n-2,n-1}^{s,t}\\
m_{n,0}^{s,t}&\cdots&m_{n,n-1}^{s,t}
\end{array}
\right|,
\end{align*}
then we have $p_n^{s,t}=-\tit_n^{s,t}/\tau_n^{s,t}$.
\begin{remark}
The four-term recurrence relation is based on the fact that
\begin{align*}
\int_0^{1} \left(P_n^{s,t}(x)+a_n^{s,t}P_{n-1}^{s,t}(x)\right)dx=0
\end{align*}
rather than $$\mathcal{L}_{s,t}\left(
P_n^{s,t}(x)+a_n^{s,t}P_{n-1}^{s,t}(x)
\right)=0.$$ Therefore, we need to introduce single moments $m_i^{s,t}$ instead of $\phi_i^{s,t}$ to express the coefficient $a_n^{s,t}$.
\end{remark}

\subsection{Discrete spectral transformation I}
The first discrete spectral transformation is related to another family of polynomials $\{Q_n^{s,t}(x)\}_{n\in\mathbb{N}}$.
Let's define monic polynomials
\begin{align*}
Q_n^{s,t}(x)=\frac{1}{\xi_n^{s,t}}\left|\begin{array}{cccc}
m_{0,1}^{s,t}&\cdots&m_{0,n}^{s,t}&1\\
\vdots&&\vdots&\vdots\\
m_{n,1}^{s,t}&\cdots&m_{n,n}^{s,t}&x^n\\
\end{array}
\right|,\quad \xi_n^{s,t}=\det\left(
m_{i,j+1}^{s,t}
\right)_{i,j=0}^{n-1}.
\end{align*}
In the theory of orthogonal polynomials, this family of polynomials is called the adjacent family of polynomials $\{P_n^{s,t}(x)\}_{n\in\mathbb{N}}$. Similar to the normalization factor $\tau_{n}^{s,t}$, we can show that 
\begin{align*}
\xi_n^{s,t}=\int_{0<x_1<\cdots<x_n<1\atop 0<y_1<\cdots<y_n<1}\prod_{i,j=1}^n \frac{x_i^sy_j^{s+1}}{x_i+y_j}\prod_{1\leq i<j\leq n}(x_i-x_j)^2(y_i-y_j)^2\prod_{i=1}^n \left(
\frac{1-x_i}{1+x_i}
\right)^t\left(
\frac{1-y_i}{1+y_i}
\right)^t
dx_idy_i,
\end{align*}
which is a non-zero partition function.
By using Jacobi determinant identity to $\tau_{n+1}^{s,t}P_{n+1}^{s,t}$ on its first and last columns and the first and last rows, we have the following discrete transformation.
\begin{proposition}\label{pro2.5}
There is a discrete transformation 
\begin{align*}
P_n^{s+1,t}(x)=\frac{1}{x}\left(
P_{n+1}^{s,t}(x)+\frac{\xi_n^{s,t}\xi_{n+1}^{s,t}}{\tau_n^{s+1,t}\tau_{n+1}^{s,t}}Q_n^{s,t}(x)
\right).
\end{align*}
\end{proposition}
Unlike the ordinary Christoffel transformation connecting $P_n^{s,t}(x)$ and $P_n^{s+1,t}(x)$, here we need to use an auxiliary polynomial $Q_n^{s,t}(x)$. In the next proposition, we demonstrate a Christoffel transformation for $Q_n^{s,t}(x)$. 
\begin{proposition}\label{pro2.6}
For auxiliary polynomials  $\{Q_n^{s,t}(x)\}_{n\in\mathbb{N}}$, we have the discrete spectral transformation
\begin{align*}
Q_n^{s,t}(x)=xP_{n-1}^{s+1,t}(x)-\frac{\xi_{n-1}^{s,t}\tau_n^{s+1,t}}{\xi_n^{s,t}\tau_{n-1}^{s+1,t}}Q_{n-1}^{s,t}(x).
\end{align*}
\end{proposition}
\begin{proof}
This transformation could also be derived by the property of $Q_n^{s,t}(x)$.
Although $Q_n^{s,t}(x)$ can't form an orthogonal system, they have the relation
\begin{align*}
\langle Q_n^{s,t}(x),y^{i}\rangle_{s,t}=0,\quad 1\leq i\leq n.
\end{align*}
By expanding $xP_{n-1}^{s+1,t}(x)$ in terms of $Q_n^{s,t}(x)$, we have 
\begin{align*}
xP_{n-1}^{s+1,t}(x)=Q_n^{s,t}(x)+\sum_{k=0}^{n-1} \zeta_{n,k}^{s,t}Q_k^{s,t}(x). 
\end{align*}
Taking the inner product $\langle\cdot,y^i\rangle_{s,t}\, (i=1,\cdots,n)$, we have
\begin{align*}
\langle xP_{n-1}^{s+1,t}(x),y^i\rangle_{s,t}=\langle P_{n-1}^{s+1,t},y^{i-1}\rangle_{s+1,t}=0,\quad 1\leq i\leq n-1,
\end{align*}
and thus
\begin{align*}
\zeta_{n,k}^{s,t}=0,\quad 0\leq k\leq n-2,\quad \zeta_{n,n-1}^{s,t}=\langle xP_{n-1}^{s+1,t},y^{n}\rangle_{s,t}\langle Q_{n-1}^{s,t},y^n\rangle_{s,t}^{-1}=\frac{\xi_{n-1}^{s,t}\tau_n^{s+1,t}}{\xi_n^{s,t}\tau_{n-1}^{s+1,t}}.
\end{align*}
\end{proof}

Using Propositions \ref{pro2.5} and \ref{pro2.6}, we can conclude the following discrete spectral transformation for $\{P_n^{s,t}(x)\}_{n\in\mathbb{N}}$.
\begin{theorem}
$(s,t)$-deformed Cauchy orthogonal polynomials  $\{P_n^{s,t}(x)\}_{n\in\mathbb{N}}$ satisfy the following spectral transformation
\begin{align}\label{spec1}
x\left(
P_n^{s+1,t}(x)+\alpha_n^{s,t}P_{n-1}^{s+1,t}(x)
\right)=P_{n+1}^{s,t}(x)+\beta_n^{s,t}P_n^{s,t}(x),
\end{align}
where 
\begin{align*}
\alpha_n^{s,t}=\frac{\xi_{n+1}^{s,t}\tau_n^{s,t}}{\tau_{n+1}^{s,t}\xi_n^{s,t}}-\frac{\xi_{n+1}^{s,t}\xi_n^{s,t}}{\tau_{n+1}^{s,t}\tau_n^{s+1,t}},\quad \beta_n^{s,t}=\frac{\xi_{n+1}^{s,t}\tau_n^{s,t}}{\tau_{n+1}^{s,t}\xi_n^{s,t}}.
\end{align*}
\end{theorem}

\subsection{Discrete spectral transformation II}
This part is devoted to the discrete evolution on index $t$. Firstly, we have the following discrete transformation by introducing a family of auxiliary polynomials $\{R_n^{s,t}(x)\}_{n\in\mathbb{N}}$. We first define monic polynomials
\begin{align*}
R_n^{s,t}(x)=\frac{(-1)^{n-1}}{\sigma_{n-1}^{s,t}}\left|\begin{array}{ccccc}
\phi^{s,t}_0&m^{s,t}_{0,0}&\cdots&m^{s,t}_{0,n-2}&1\\
\phi^{s,t}_1&m^{s,t}_{1,0}&\cdots&m^{s,t}_{1,n-2}&x\\
\vdots&\vdots&&\vdots&\vdots\\
\phi^{s,t}_n&m^{s,t}_{n,0}&\cdots&m^{s,t}_{n,n-2}&x^n
\end{array}
\right|,
\end{align*}
where \begin{align*}
\sigma_n^{s,t}=\left|
\begin{array}{cccc}
m_{0,0}^{s,t}&\cdots&m_{0,n-1}^{s,t}&\phi_0^{s,t}\\
m_{1,0}^{s,t}&\cdots&m_{1,n-1}^{s,t}&\phi_1^{s,t}\\
\vdots&&\vdots&\vdots\\
m_{n,0}^{s,t}&\cdots&m_{n,n-1}^{s,t}&\phi_n^{s,t}
\end{array}
\right|.
\end{align*}
It should be noted that $\sigma_n^{s,t}$ given above is different from $\tilde{\sigma}_n^{s,t}$ given in \eqref{sigma}. 
In fact, $\sigma_n^{s,t}$ is defined by
 $\sigma_n^{s,t}=\tau_n^{s,t}\mathcal{L}_{s,t}(P_n^{s,t})$.
By using auxiliary polynomials $\{R_n^{s,t}(x)\}_{n\in\mathbb{N}}$, we can claim the following proposition.
\begin{proposition}\label{prop2.8}
There is a discrete transformation
\begin{align}\label{dt-1}
P_n^{s,t+1}(x)=R_n^{s,t}(x)-d_n^{s,t}P_{n-1}^{s,t+1}(x),
\end{align}
where 
\begin{align*}
d_n^{s,t}=-\frac{\sigma_n^{s,t}\tau_{n-1}^{s,t+1}}{\sigma_{n-1}^{s,t}\tau_n^{s,t+1}}.
\end{align*}
\end{proposition}
\begin{proof}
By substituting $m_{i,j}^{s,t+1}=m_{i,j}^{s,t}-\phi_i^{s,t}\phi_j^{s,t}$, we know
\begin{align*}
P_n^{s,t+1}(x)=\frac{1}{\tau_n^{s,t+1}}\left|\begin{array}{ccccc}
1&\phi^{s,t}_0&\cdots&\phi^{s,t}_{n-1}&0\\
\phi^{s,t}_0&m^{s,t}_{0,0}&\cdots&m^{s,t}_{0,n-1}&1\\
\vdots&\vdots&&\vdots&\vdots\\
\phi^{s,t}_n&m^{s,t}_{n,0}&\cdots&m^{s,t}_{n,n-1}&x^n
\end{array}
\right|.
\end{align*}
An application of Jacobi identity to the above determinant for the first and last rows and last two columns gives the result. 
\end{proof}
In fact, $R_n^{s,t}(x)$ can be uniquely determined by relations
\begin{align}\label{propr}
\langle R_n^{s,t}(x), y^i\rangle_{s,t}=0,\quad 0\leq i\leq n-2,\qquad \mathcal{L}_{s,t}(R_n^{s,t})=0.
\end{align}
By making the use of $R_n^{s,t}(x)$, we have the following discrete spectral transformation.
\begin{theorem}
We have the discrete spectral transformation
\begin{align}\label{trans2}
P_n^{s,t+1}(x)+d_n^{s,t}P_{n-1}^{s,t+1}(x)=P_n^{s,t}(x)+e_{n}^{s,t}P_{n-1}^{s,t}(x),
\end{align}
where 
\begin{align*}
e_n^{s,t}=-\frac{\sigma_n^{s,t}\tau_{n-1}^{s,t}}{\sigma_{n-1}^{s,t}\tau_{n}^{s,t}}.
\end{align*}
\end{theorem}
\begin{proof}
It should be noted that $\{P_n^{s,t}(x)\}_{n\in\mathbb{N}}$ form a set of basis for polynomials, and thus
\begin{align*}
R_n^{s,t}(x)=P_n^{s,t}(x)+\sum_{k=0}^{n-1}\alpha_k^{s,t}P_k^{s,t}(x).
\end{align*}
The property of $R_n^{s,t}$ (c.f equation \eqref{propr}) leads to the result
\begin{align*}
\alpha_k^{s,t}=0,\quad 0\leq k\leq n-2, \qquad \alpha_{n-1}^{s,t}=\langle R_n^{s,t}(x),y^{n-1}\rangle_{s,t} (h_{n-1}^{s,t})^{-1}.
\end{align*}
Moreover, by comparing the coefficient, we obtain
\begin{align*}
R_n^{s,t}(x)=P_n^{s,t}(x)+e_{n}^{s,t}P_{n-1}^{s,t}(x), 
\end{align*}
which results in the result by using Proposition \ref{prop2.8}.
\end{proof}
\subsection{Algebraic construction of an integrable hierarchy}
In this part, we first reformulate the above discrete spectral transformations, and then give a discrete integrable hierarchy by using compatibility conditions. Let $\Lambda$ be a shift operator and $\Phi^{s,t}$ be a semi-infinite column vector generated by polynomials $P_0^{s,t}(x)$, $P_1^{s,t}(x)$, $\cdots$. With these notations, the four-term recurrence relation could be reformulated as
\begin{align*}
&x\Phi^{s,t}=\mathcal{L}^{s,t}\Phi^{s,t},\\
&\mathcal{L}=(I+A^{s,t}\Lambda^{-1})^{-1}(\Lambda+A^{s,t}-B^{s,t}-C^{s,t}\Lambda^{-1}+A^{s,t}\Lambda^{-1}B^{s,t}-A^{s,t}\Lambda^{-1}C^{s,t}\Lambda^{-1}),
\end{align*}
where 
$I$ is the identity matrix and $A^{s,t}=\text{diag}(a_0^{s,t},a_1^{s,t},\cdots)$, $B^{s,t}=\text{diag}(b_0^{s,t},b_1^{s,t},\cdots)$ and $C^{s,t}=\text{diag}(c_0^{s,t},c_1^{s,t},\cdots)$. Similarly, discrete spectral transformations \eqref{spec1} and \eqref{trans2} could be written as
\begin{align*}
x\Phi^{s+1,t}=\mathcal{N}^{s,t}\Phi^{s,t},\quad \Phi^{s,t+1}=\mathcal{M}^{s,t}\Phi^{s,t},
\end{align*}
where 
\begin{align*}
\mathcal{N}^{s,t}=(I+\alpha^{s,t}\Lambda^{-1})^{-1}(\Lambda+\beta^{s,t}),\quad 
\mathcal{M}^{s,t}=(I+E^{s,t}\Lambda^{-1})^{-1}(I+D^{s,t}\Lambda^{-1}),
\end{align*}
and 
\begin{align*}
&\alpha^{s,t}=\text{diag}(\alpha_0^{s,t},\alpha_1^{s,t},\cdots),\quad
\beta^{s,t}=\text{diag}(\beta_0^{s,t},\beta_1^{s,t},\cdots),\\
&D^{s,t}=\text{diag}(d_0^{s,t},d_1^{s,t},\cdots),\quad E^{s,t}=\text{diag}(e_0^{s,t},e_1^{s,t},\cdots).
\end{align*}
By using compatibility conditions, we have the following integrable hierarchy.
\begin{theorem}
There is an integrable hierarchy given by 
\begin{align}\label{cp}
\mathcal{M}^{s+1,t}\mathcal{N}^{s,t+1}=\mn^{s,t}\mm^{s,t},\quad \ml^{s+1,t}\mn^{s,t}=\mn^{s,t}\ml^{s,t},\quad \mm^{s,t}\ml^{s,t+1}=\ml^{s,t}\mm^{s,t}.
\end{align}
\end{theorem}
By recognizing that 
\begin{align*}
&\ml^{s,t}=\Lambda-B^{s,t}-C^{s,t}\Lambda^{-1}+O(\Lambda^{-2}),\\
&\mm^{s,t}=I+(D^{s,t}-E^{s,t})\Lambda^{-1}+E^{s,t}\Lambda^{-1}(E^{s,t}-D^{s,t})\Lambda^{-1}+O(\Lambda^{-3}),\\
&\mn^{s,t}=\Lambda+(\beta^{s,t}-\alpha^{s,t})+\alpha^{s,t}\Lambda^{-1}(\alpha^{s,t}-\beta^{s,t})+O(\Lambda^{-2}),
\end{align*}
we can obtain the following nonlinear difference equation from \eqref{cp} which reads
\begin{align*}
&g_n^{s+1,t}+f_n^{s,t+1}=f_n^{s,t}+g_{n+1}^{s,t},\\
&f_{n+1}^{s,t}-b_n^{s+1,t}=f_n^{s,t}-b_{n+1}^{s,t},\\
&g_n^{s,t}-b_n^{s,t+1}=g_{n+1}^{s,t}-b_n^{s,t}\\
&(g_n^{s+1,t}-\alpha_n^{s,t+1})f_{n-1}^{s,t+1}-e_n^{s+1,t}g_{n-1}^{s+1,t}=(g_{n-1}^{s,t}-\alpha_n^{s,t})f_{n-1}^{s,t}-e_n^{s,t}g_{n-1}^{s,t},\\
&c_n^{s+1,t}+b_n^{s+1,t}f_n^{s,t}-\alpha_{n+1}^{s,t}f_n^{s,t}=c_{n+1}^{s,t}-f_n^{s,t}b_n^{s,t}-\alpha_n^{s,t}f_{n-1}^{s,t},\\
&c_{n-1}^{s,t+1}-g_n^{s,t}b_{n-1}^{s,t+1}+e_n^{s,t}g_{n-1}^{s,t}=c_{n-1}^{s,t}-g_{n-1}^{s,t}b_{n-1}^{s,t}+e_{n-1}^{s,t}g_{n-1}^{s,t},
\end{align*}
where we use notations $f_n^{s,t}=\beta_n^{s,t}-\alpha_n^{s,t}$ and $g_n^{s,t}=d_n^{s,t}-e_n^{s,t}$.
These equations with unknown variables $b_n^{s,t},\,c_n^{s,t},\,d_n^{s,t},\,e_n^{s,t},\alpha_n^{s,t},\,\beta_n^{s,t}$ form an integrable system from compatibility conditions. However, the determinant expressions of these variables are only related to $\tau_n^{s,t},\,\tit_n^{s,t},\,\xi_n^{s,t},\,\sigma_n^{s,t}$. It means that there are some variables could be eliminated, from which a simpler equation could be obtained. In the subsequent, we would derive an equivalent equation expressed by $\tau_n^{s,t}$ and show that the corresponding discrete integrable system is exactly the discrete CKP equation.
\begin{remark}
Note that in the semi-discrete C-Toda equation, the nonlinear variables are expressed by $a_n\,b_n$ and $c_n$. Therefore, we can see the differences between continuous time evolutions and discrete time evolutions in the discrete setting. 
\end{remark}

\section{The discrete CKP equation}\label{sec3}
In last section, we showed an overdetermined system from the compatibility conditions of spectral transformations. In this section, we show that there are some simpler relations between variables which could be obtained directly from discrete spectral transformations.
The first relation between $\tau_n^{s,t}$ and $\xi_n^{s,t}$ comes from the discrete spectral transformation \eqref{spec1}.
\begin{proposition}
We have the formula
\begin{align}\label{eq1}
\tau_{n+1}^{s,t}\tau_{n-1}^{s+1,t}=\tau_n^{s,t}\tau_n^{s+1,t}-(\xi_n^{s,t})^2.
\end{align}
\end{proposition}
\begin{proof}
By acting an inner product $\langle \cdot, y^{n}\rangle_{s,t}$ on both sides of \eqref{spec1}, we could get $
\alpha_n^{s,t}h_{n-1}^{s+1,t}=\beta_n^{s,t}h_n^{s,t}
$. Then the equation \eqref{eq1} is a direct result by substituting $\alpha^{s,t}_n$ and $\beta_n^{s,t}$ into this formula. 
\end{proof}
\begin{remark}
By using this proposition, $\alpha_n^{s,t}$ in \eqref{spec1} could be simplified as
\begin{align*}
\alpha_n^{s,t}=\frac{\xi_{n+1}^{s,t}\tau_{n-1}^{s+1,t}}{\xi_n^{s,t}\tau_n^{s+1,t}}.
\end{align*}
\end{remark}
Moreover, we have the following bilinear relations from discrete spectral transformations \eqref{spec1} and \eqref{trans2}.
\begin{proposition}
We have equations
\begin{subequations}
\begin{align}
&\tau_{n+1}^{s,t+1}\tau_{n}^{s,t}-\tau_{n}^{s,t+1}\tau_{n+1}^{s,t}=(\sigma_{n}^{s,t})^2,\label{3.2a}\\
&\tau_n^{s,t}\tau_{n}^{s+1,t+1}-\tau_{n+1}^{s,t}\tau_{n-1}^{s+1,t+1}=\xi_n^{s,t+1}\xi_n^{s,t}-\sigma_{n-1}^{s+1,t}\sigma_n^{s,t}.\label{3.2b}
\end{align}
\end{subequations}
\end{proposition}
\begin{proof}
From equation \eqref{spec1}, if we act the linear functional $\mathcal{L}_{s,t}$ on both sides, and note that
\begin{align*}
\mathcal{L}_{s,t}\left(
xP_n^{s+1,t}(x)
\right)=\mathcal{L}_{s+1,t}\left(P_n^{s+1,t}(x)\right)= \frac{\sigma_n^{s+1,t}}{\tau_n^{s+1,t}},\quad \mathcal{L}_{s,t}\left(P_n^{s,t}(x)\right)= \frac{\sigma_n^{s,t}}{\tau_n^{s,t}}, 
\end{align*}
then the equality could be written in terms of a trilinear equation
\begin{align}\label{tri1}
\sigma_n^{s+1,t}\xi_n^{s,t}\tau_{n+1}^{s,t}+\xi_{n+1}^{s,t}\sigma_{n-1}^{s+1,t}\tau_{n+1}^{s,t}=\sigma_{n+1}^{s,t}\tau_n^{s+1,t}\xi_n^{s,t}+\xi_{n+1}^{s,t}\sigma_n^{s,t}\tau_n^{s+1,t}.
\end{align}
Such a trilinear form could be decomposed into two bilinear equations\footnote{It should be remarked that there are many different bilinear relations could be obtained from the trilinear equation. For example, we have alternative formulas
\begin{align*}
&\sigma_n^{s+1,t}\tau_{n+1}^{s,t}-\sigma_{n+1}^{s,t}\tau_n^{s+1,t}=\xi_{n+1}^{s,t}\psi_n^{s,t}\\
&\sigma_n^{s,t}\tau_n^{s+1,t}-\sigma_{n-1}^{s+1,t}\tau_{n+1}^{s,t}=\xi_n^{s,t}\psi_n^{s,t}.
\end{align*}
However, these formulas are useless to derive a simpler equation.}
\begin{subequations}
\begin{align}
&\sigma_{n+1}^{s,t}\xi_n^{s,t}+\xi_{n+1}^{s,t}\sigma_n^{s,t}=\tau_{n+1}^{s,t}\psi_n^{s,t},\label{3.3a}\\
&\sigma_{n}^{s+1,t}\xi_n^{s,t}+\sigma_{n-1}^{s+1,t}\xi_{n+1}^{s,t}=\tau_n^{s+1,t}\psi_{n}^{s,t}\label{3.3b}
\end{align}
\end{subequations}
by introducing a variable 
\begin{align*}
\psi_n^{s,t}=\left|\begin{array}{cccc}
m_{0,1}^{s,t}&\cdots&m_{0,n}^{s,t}&\phi_0^{s,t}\\
m_{1,1}^{s,t}&\cdots&m_{1,n}^{s,t}&\phi_1^{s,t}\\
\vdots&&\vdots&\vdots\\
m_{n,1}^{s,t}&\cdots&m_{n,n}^{s,t}&\phi_{n}^{s,t}
\end{array}
\right|.
\end{align*}
On the other hand, from equation \eqref{trans2} and by noting that
\begin{align*}
P_n^{s,t}(0)=(-1)^{n}\frac{\xi_{n}^{s,t}}{\tau_n^{s,t}},
\end{align*}
we could get another trilinear equation
\begin{align}\label{tri2}
\tau_n^{s,t}\sigma_{n-1}^{s,t}\xi_{n}^{s,t+1}+\sigma_n^{s,t}\xi_{n-1}^{s,t+1}\tau_n^{s,t}=\xi_n^{s,t}\tau_{n}^{s,t+1}\sigma_{n-1}^{s,t}+\sigma_n^{s,t}\xi_{n-1}^{s,t}\tau_n^{s,t+1}.
\end{align}
We can decompose this trilinear equation into
\begin{subequations}
\begin{align}
&\tau_n^{s,t+1}\xi_n^{s,t}-\tau_n^{s,t}\xi_n^{s,t+1}=-\sigma_n^{s,t}\psi_{n-1}^{s,t},\label{3.4a}\\
&\tau_n^{s,t}\xi_{n-1}^{s,t+1}-\tau_n^{s,t+1}\xi_{n-1}^{s,t}=-\sigma_{n-1}^{s,t}\psi_{n-1}^{s,t}.\label{3.4b}
\end{align}
\end{subequations}
By multiplying $\psi_{n}^{s,t}$ on \eqref{3.3a}, and taking a substitution of \eqref{3.4a} and \eqref{3.4b}, we get the formula
\begin{align*}
\xi_{n+1}^{s,t+1}\xi_{n}^{s,t}-\xi_{n}^{s,t+1}\xi_{n+1}^{s,t}=(\psi_{n}^{s,t})^2.
\end{align*}
It should be noted that 
\begin{align}\label{subs}
\begin{aligned}
&\xi_n^{s,t}=\tau_n^{s,t}|_{m_{i,j}^{s,t}\to m_{i,j+1}^{s,t}},\qquad \psi_n^{s,t}=\sigma_n^{s,t}|_{m^{s,t}_{i,j}\to m^{s,t}_{i,j+1}}
\end{aligned}
\end{align} 
and thus we can get equation \eqref{3.2a} by replacing $\xi_n^{s,t}$ and $\psi_n^{s,t}$ by $\tau_n^{s,t}$ and $\sigma_n^{s,t}$ respectively.

If we multiply $\psi_{n-1}^{s+1,t}$ on \eqref{3.3b}, then we can get
\begin{align*}
\tau_n^{s+1,t}\left(
\xi_n^{s+1,t+1}\xi_n^{s,t}-\xi_{n+1}^{s,t}\xi_{n-1}^{s+1,t+1}-\psi_{n-1}^{s+1,t}\psi_n^{s,t}
\right)+\tau_n^{s+1,t+1}\left(
\xi_{n+1}^{s,t}\xi_{n-1}^{s+1,t}-\xi_n^{s,t}\xi_n^{s+1,t}
\right)=0.
\end{align*}
From \eqref{eq1}, it is known that
\begin{align*}
\xi_{n+1}^{s,t}\xi_{n-1}^{s+1,t}-\xi_n^{s,t}\xi_n^{s+1,t}=\tau_n^{s+1,t}\hat{\tau}_n^{s,t},
\end{align*}
where $\hat{\tau}_n^{s,t}=\det(m_{i,j+2}^{s,t})_{i,j=0}^{n-1}$.
Thus we get 
\begin{align*}
\xi_n^{s+1,t+1}\xi_n^{s,t}-\xi_{n+1}^{s,t}\xi_{n-1}^{s+1,t+1}=\tau_n^{s+1,t+1}\hat{\tau}_n^{s,t}-\psi_{n-1}^{s+1,t}\psi_n^{s,t}.
\end{align*}

If we take the substitution \eqref{subs} backwards and note that when shift $m_{i,j+1}^{s,t}$ to $m_{i,j}^{s,t}$, $\hat{\tau}_n^{s,t}$ turns into $\xi_n^{s,t}$ and $\tau_n^{s+1,t+1}$ turns into $\xi_n^{s,t+1}$, then we get the equation \eqref{3.2b}.
\end{proof}
By using equations \eqref{eq1}, \eqref{3.2a} and \eqref{3.2b} and eliminating $\sigma_n^{s,t}$ and $\xi_n^{s,t}$, there is a discrete equation governed by $\tau_n^{s,t}$, which gives the following theorem.
\begin{theorem}
The normalization factor $\tau_n^{s,t}$ satisfies the following discrete CKP equation
\begin{align}\label{dckp}
\begin{aligned}
&4\left(
\tau_{n}^{s+1,t}\tau_n^{s,t}-\tau_{n+1}^{s,t}\tau_{n-1}^{s+1,t}
\right)\left(
\tau_n^{s+1,t+1}\tau_n^{s,t+1}-\tau_{n+1}^{s,t+1}\tau_{n-1}^{s+1,t+1}
\right)\\
&\qquad=\left(
\tau_n^{s+1,t}\tau_n^{s,t+1}+\tau_n^{s+1,t+1}\tau_n^{s,t}-\tau_{n+1}^{s,t+1}\tau_{n-1}^{s+1,t}-\tau_{n+1}^{s,t}\tau_{n-1}^{s+1,t+1}
\right)^2.
\end{aligned}
\end{align}
\end{theorem}
This theorem implies that the discrete CKP equation admits the following multiple integral solution
\begin{align*}
\tau_n^{s,t}=\int_{0<x_1<\cdots<x_n<1\atop 0<y_1<\cdots<y_n<1}\prod_{i,j=1}^n \frac{x_i^sy_j^s}{x_i+y_j}\prod_{1\leq i<j\leq n}(x_i-x_j)^2(y_i-y_j)^2\prod_{i=1}^n \left(
\frac{1-x_i}{1+x_i}
\right)^t\left(
\frac{1-y_i}{1+y_i}
\right)^t
dx_idy_i.
\end{align*}
Moreover, from bilinear relations, we have the following corollary.
\begin{coro}
The discrete CKP equation \eqref{dckp} admits the following determinant solution
\begin{align*}
\tau_n^{s,t}=\det(m_{i,j}^{s,t})_{i,j=0}^{n-1},
\end{align*}
where 
\begin{align*}
m_{i,j}^{s+1,t}=m_{i+1,j+1}^{s,t},\quad m_{i,j}^{s,t+1}=m_{i,j}^{s,t}-\phi_i^{s,t}\phi_j^{s,t}
\end{align*}
and $\phi_i^{s+1,t}=\phi_{i+1}^{s,t}$.
\end{coro}

In fact, there are some other simpler relations satisfied by these variables, whose compatibility condition would lead to the discrete CKP equation as well. 
\begin{proposition}
We have the following relations
\begin{subequations}
\begin{align}
&\tau_{n+1}^{s,t}\tau_{n-1}^{s+1,t}=\tau_n^{s,t}\tau_n^{s+1,t}-(\xi_n^{s,t})^2,\label{e1}\\
&\tau_n^{s,t+1}\tau_{n-1}^{s,t}=\tau_n^{s,t}\tau_{n-1}^{s,t+1}-(\sigma_{n-1}^{s,t})^2,\label{e2}\\
&\tau_n^{s,t+1}\tau_{n-1}^{s+1,t}=\tau_n^{s,t}\tau_{n-1}^{s+1,t+1}-(\psi_{n-1}^{s,t})^2,\label{e3}\\
&\tau_n^{s,t+1}\xi_{n-1}^{s,t}=\tau_n^{s,t}\xi_{n-1}^{s,t+1}-\psi_{n-1}^{s,t}\sigma_{n-1}^{s,t}.\label{e4}
\end{align}
\end{subequations}
\end{proposition}
\begin{proof}
It should be noted that \eqref{e1}, \eqref{e2} and \eqref{e4} have been derived from spectral transformation of deformed Cauchy bi-orthogonal polynomials in \eqref{eq1}, \eqref{3.2a} and \eqref{3.4a}. However, they can be directly verified by Jacobi's determinant identity. The first equation could be verified by using Jacobi's identity for the first and last rows and columns to the determinant representation of $\tau_{n+1}^{s,t}$, and the second (respectively, the third and the last) equations are verified by using Jacobi's identities for the $(1,n+1)$-rows and $(1,n+1)$-columns (respectively, the $(1,2)$-rows and $(1,2)$-columns, and $(1,2)$-rows and $(1,n+1)$-columns) to the determinant $\tau_n^{s,t+1}$.
\end{proof}
\begin{remark}
Despite of using determinant identities, these relations could be verified by using Pfaffians as well. In \cite{li20}, we give some similar formulas for $\tau_n^{s,t}$ with different kinds of discrete evolutions. These equations are equivalent under an affine transformation of index $n,\,s$ and $t$. 
\end{remark}

We should give some remarks at the end of this part. Firstly, regarding with the equation \eqref{dckp}, 
it was firstly found by Kashaev from the star-triangle transformation in the Ising model based on the idea of nonlocal Yang-Baxter relation \cite{kashaev95}. Later on, Schief found the same equation  from the perspective of B\"acklund transformation, and viewed it as a superposition formula for the continuous CKP equation \cite{schief03}. Some geometric explanations of such an equation were given in \cite{doliwa10,bobenko15}.  Besides, a soliton solution of discrete CKP equation was given by Fu and Nijhoff by direct linearization \cite{fu17}, and general solutions were constructed by a fermion approach \cite{harnad23}. 
In recent years, the combinatoric explanation for the discrete CKP equation was also considered, as analogy of discrete KP (Hirota-Miwa) equation and discrete BKP (Miwa) equation. In the paper of Kenyon and Pemantle \cite{kenyon15}, the authors showed that discrete CKP equation was a special case of hexahedron recurrence, shown to have Laurent property by using a superurban renewal. 
 In fact, equations \eqref{e1}-\eqref{e3} has also been obtained \cite[Prop. 6.6]{kenyon15} by interpreting the $\tau$-function $\tau_n^{s,t}$ as vertex variable and auxiliary functions $\sigma_n^{s,t}$, $\xi_n^{s,t}$ and $\psi_n^{s,t}$ as face variables.

\section{Concluding remarks}
In this paper, we derived the discrete CKP equation by using deformed Cauchy orthogonal polynomials. By introducing a deformed Jacobi weight, we demonstrated the special Cauchy orthogonal polynomials satisfied two different spectral transformations, which are compatible with four-term recurrence relation satisfied by general Cauchy bi-orthogonal polynomials. Therefore, we obtain some algebraic systems from Lax pairs. It is our belief that these complicated nonlinear system should have a simpler form, since there is a simple $\tau$-function expressed by determinant. Then we use bilinear equations to simplify the above-mentioned algebraic equations, and find the equation governed by a single $\tau$-function. This is nothing but a discrete CKP equation. 

However, it is still of interest to know more about the Jacobi-deformed Cauchy orthogonal polynomials and related matrix model. Although we have the multiple integral representation for the normalization and polynomials, we'd like to know their exact expressions in terms of hypergeometric function and their limiting behavior when the degree of polynomials tend to infinity. 
\section*{Acknowledgement}
This work is partially funded by grants (Grant No. NSFC 12371251, 12175155). G.F. Yu was also partially supported by the National Key Research and Development Program of China (Grant No. 2024YFA1014101), Shanghai Frontier Research Institute for Modern Analysis and the Fundamental Research Funds for the Central Universities.

\section*{Data Availability Statement}
There is no data needed in this research.


\begin{thebibliography}{1}

\bibitem{harnad23}
S. Arthamonov, J. Harnad and J. Hurtubise.
Lagrangian Grassmannians, CKP Hierarchy and Hyperdeterminantal Relations.
\emph{Comm. Math. Phys.}, 401 (2023), 1337-1381.

\bibitem{bertola09}
M. Bertola, M. Gekhtman, and J. Szmigielski.
The Cauchy two-matrix model. \emph{Commun. Math. Phys.}, 287 (2009),
983–1014.

\bibitem{bertola10}
M. Bertola, M. Gekhtman, and J. Szmigielski.
Cauchy biorthogonal polynomials. \emph{J. Approx. Theory}, 162 (2010),
832–867.

\bibitem{bertola14}
M. Bertola, M. Gekhtman, and J. Szmigielski.
Cauchy-Laguerre two-matrix model and the Meijer-G random point field.
\emph{Comm. Math. Phys.,} 326 (2014), 111-144.

\bibitem{bobenko15}
A. Bobenko and W. Schief.
Circle complexes and the discrete CKP equation.
\emph{Int. Math. Res. Not.}, 5 (2017), 1504-1561.

\bibitem{chang21}
X. Chang, S. Li, S. Tsujimoto and G. Yu. Two-parameter generalizations of Cauchy bi-orthogonal polynomials and integrable lattices. \emph{J. Nonlinear Sci.}, 31 (2021), paper 30.

\bibitem{doliwa10}
A. Doliwa.
The C-(symmetric) quadrilateral lattice, its transformations and the algebro-geometric construction.
\emph{J. Geom. Phys.}, 60 (2010), 690-707.

\bibitem{forrester21}
P. Forrester and S. Li. Fox H-kernel and theta-deformation of the Cauchy two-matrix model and Bures ensemble. \emph{Int. Math. Res. Not.}, 8 (2021), 5791–5824.

\bibitem{fu17}
W. Fu and F. Nijhoff. 
Direct linearizing transform and three-dimensional discrete integrable systems: the lattice AKP, BKP and CKP equations.
\emph{Proc. A}, 473 (2017), 20160915.

\bibitem{kashaev95}
R. Kashaev. On discrete three-dimensional equations associated
with the local Yang-Baxter relation. \emph{Lett. Math. Phys.}, 33 (1996), 389-397.

\bibitem{kenyon15}
R. Kenyon and  R. Pemantle. Double-dimers, the Ising model and the
hexahedron recurrence. \emph{J. Comb. Theory A.}, 137 (2016), 27-63.

\bibitem{li19}
C. Li and S. Li. The Cauchy two-matrix model, C-Toda lattice and CKP hierarchy, \emph{J. Nonlinear Sci.}, 29 (2019), 3-27.

\bibitem{li20}
S. Li. 
Discrete integrable systems and condensation algorithms for Pfaffians.
arXiv: 2006.06221, 2020.

\bibitem{li23} 
S. Li, Y. Shi, G. Yu and J. Zhao. Matrix-valued Cauchy bi-orthogonal polynomials and a novel noncommutative integrable lattice, arXiv: 2212.14512.

\bibitem{lundmark05}
H. Lundmark and J. Szmigielski.
Degasperis–Procesi peakons and the discrete cubic string.
Int. Math. Res. Papers., 2 (2005), 53-116.


\bibitem{miki11}
H. Miki, and S. Tsujimoto. Cauchy biorthogonal polynomials and discrete integrable systems. \emph{J. Nonlinear
Syst. Appl.}, 2 (2011), 195–199.

\bibitem{schief03}
W. Schief. Lattice geometry of the discrete Darboux, KP, BKP
and CKP equations. Menelaus’ and Carnot’s theorems. \emph{J. Nonlinear Math.
Phys.}, 10, Supplement 2, (2003) 194-208.
	
\end{thebibliography}
\end{document}